%% file: main.tex
\documentclass{ceurart}

\usepackage{amssymb}
\usepackage{graphicx}
\usepackage{enumitem}
\setlist[enumerate]{noitemsep, topsep=2pt}
\setlist[description]{noitemsep, topsep=2pt}
\usepackage{mathpazo}   % Palatino text + matched math
\usepackage{listings}   % verbatim Lean statements in the appendix
\AtBeginDocument{}

\DeclareMathOperator*{\argmin}{arg\,min}
\newcommand{\rel}[1]{\mathit{#1}}
\newcommand{\aux}{\mathrm{Aux}}
\newcommand{\Huniv}{\mathfrak{U}}
\newcommand{\Hbase}{\mathfrak{H}}
\newcommand{\Hsub}{\mathfrak{B}}
\newcommand{\Mod}{\mathfrak{M}}

\newtheorem{theorem}{Theorem}
\newtheorem{lemma}{Lemma}
\newproof{proof}{Proof}

\begin{document}

%% Rights management information (CC-BY is the CEUR-WS default).
\copyrightyear{2026}
\copyrightclause{Copyright for this paper by its authors.
  Use permitted under Creative Commons License Attribution 4.0
  International (CC BY 4.0).}

\conference{Datalog 2.0 2026: 6th International Workshop on the Resurgence of Datalog
  in Academia and Industry, September 7, 2026, Klagenfurt, Austria}

\title{Towards Datalog on Quantum Annealers: Compiling Recursive Logic
Programs with Bottom-up Semantics to 2-local Ising Models}

\author[1,2]{Bruno Rucy Carneiro Alves de Lima}[%
  email=bruno.rucy.carneiro.alves.de.lima@ut.ee,
  orcid=0000-0002-2679-6639
]
\cormark[1]
\author[1]{Victor Henrique Cabral Pinheiro}[%
  email=victor.pinheiro@ut.ee,
  orcid=0000-0003-1112-2620
]
\author[1]{Evgenii Dolzhkov}[%
  email=evgenii.dolzhkov@ut.ee,
]
\author[3]{Joseph Haske}[%
  email=joseph.l.haske@gmail.com,
]

\address[1]{University of Tartu, Tartu, Estonia}
\address[2]{Lucidarium Systems, \url{https://lucidarium.systems}}
\address[3]{Independent Researcher}

\cortext[1]{Corresponding author.}

\begin{abstract}
Quantum annealers solve problems by finding the lowest-energy (ground) state of a programmable
physical system, a 2-local Ising model, whose energy function is the Hamiltonian. We compile
recursive Datalog programs into such models so that the ground state projects onto the
program's minimal Herbrand model.
The compiler has four stages: binarization, grounding, reduction to a Min-Ones SAT formula,
and Ising encoding. Each rule becomes an energy penalty on the one assignment that violates
it, and a small uniform
cost on every true atom selects the minimal model. We contribute both in theory and in
practice with per-stage correctness lemmas and a correspondence theorem, verified in Lean~4, establishing
that the ground state of the compiled model projects onto the program's minimal Herbrand
model. We
map the compiled models onto the topologies of commercial annealers and characterize,
under classical and simulated-quantum annealing, whether and when that certified ground state is attained.
\end{abstract}

\begin{keywords}
  Datalog \sep
  Quantum Annealing \sep
  Ising Model \sep
  Logic Programming \sep
  Bottom-up Evaluation \sep
  Minimal Model
\end{keywords}

\maketitle

\section{Introduction}

A quantum annealer~\cite{kadowaki_nishimori1998,johnson2011} solves one kind of problem:
given an energy function over binary
variables, find an assignment of minimum energy. Programming
the machine means choosing that energy function, its \emph{Hamiltonian}.

Nonrecursive, query-directed (top-down) logic programs, specifically
a subset of Prolog~\cite{colmerauer1996}, have been compiled
to \emph{2-local Ising models}~\cite{ising1925}: energy functions over $\pm 1$-valued
variables, called \emph{spins}, whose
interactions are at most pairwise, the form that D-Wave's annealers
minimize~\cite{pakin2018}.

A Datalog~\cite{ceri1989datalog} program has a least model~\cite{vanEmdenKowalski1976}, and an Ising model has
least-energy (ground) states. We show that the two minima can be made to coincide: the least
model of a program becomes the ground state of a compiled Ising model. Bottom-up
evaluation stays tractable in data complexity, so the optimum of every compiled
instance is known exactly and solvers are measured against it, not against a sampler's
estimate.

We compile bottom-up \emph{recursive} Datalog to 2-local Ising
models. Our contributions are:
\begin{enumerate}
\item a compiler in which each rule becomes an energy penalty on the one assignment that
violates it, disjunction
requires no dedicated construct, and minimality follows from a single uniform cost on every true atom
\item per-stage correctness lemmas and a correspondence theorem establishing that
the ground state of the compiled model projects onto the program's minimal Herbrand model,
with proofs in Lean~4~\cite{lean4}
\item a reference implementation in Python that compiles arbitrary programs and exactly
solves instances small enough for exhaustive verification, with an empirical
evaluation of its hardware footprint and of solver recovery under
classical and simulated-quantum annealing.
\end{enumerate}

\noindent\textbf{Related work.} Before quantum hardware, propositional satisfiability was
already modelled as energy minimization: formulas map to quadratic energy functions whose minima
are the satisfying assignments~\cite{pinkas1991}, and logic
programs were encoded so that the minima are their models, with no preference among
them~\cite{wanabdullah1992}.
On annealers, Ising formulations are catalogued for many NP
problems~\cite{lucas2014}, and SAT and MaxSAT encodings target sparse
annealer hardware, with the maximum-weight assignments at minimum
energy~\cite{bian2020}. A nonrecursive fragment of Prolog compiles to Hamiltonians whose
ground states are the answers to a single query~\cite{pakin2018}. Answer Set
Programming~\cite{gelfond1988} has also been paired with
quantum computation, though by Grover search~\cite{grover1996} on gate-based machines, with
no energy
landscape involved~\cite{romanello2024}. On the classical side, computing the
minimal models of a propositional theory is a well-studied problem, with efficient algorithms
near the Horn fragment~\cite{benEliyahuDechter1996}.

In the energy encodings above, the bottom of the landscape holds a solution set. Our
propositional objective is itself standard: the compiler reduces the ground program to
Min-Ones SAT (satisfiability with the fewest true atoms), an instance of the weighted MaxSAT that the encodings
above accept~\cite{bian2020}. What separates the compiler is what sits before and after that
objective. It starts from recursive Datalog, whose rules carry variables where the encodings
above are propositional, and makes grounding an explicit, model-preserving stage, so the
ground state projects onto the query-independent least Herbrand model of the source
program. That stage is deliberately simple. The optimized grounders of Answer Set
Programming's ground-and-solve architecture would shrink it and apply here unchanged, but the
contribution is the compilation target, not the grounder. For a
definite program that model is the unique model of smallest cardinality, so one uniform
$\varepsilon$ cost on every true atom selects it, with no per-clause weights, and enforces
the founded, recursive semantics, in which every true atom rests on a derivation from the
facts. Every stage, through the ground-state
correspondence, is verified in Lean~4.

\section{Background}\label{sec:background}

\subsection{Datalog}

\noindent\textbf{Syntax.} A Datalog program $P$ consists of \emph{rules}
$h \leftarrow b_1, \ldots, b_n$ where the head $h$ and each body atom $b_i$ are
predicates applied to terms. Terms are \emph{variables} ($X, Y$) or
\emph{constants} ($a, b$). Every rule is \emph{safe}: each head variable also occurs in the
body. A rule is \emph{recursive} when its head predicate
occurs in a body, and a \emph{ground fact} is a rule with no variables and an
empty body, such as $\rel{Edge}(a,b)$. The transitive closure program used
throughout is
\[
\rel{Reach}(X,Y) \leftarrow \rel{Edge}(X,Y)
\qquad
\rel{Reach}(X,Z) \leftarrow \rel{Edge}(X,Y), \rel{Reach}(Y,Z).
\]

\noindent\textbf{Semantics.} Given an \emph{extensional database} (EDB)
$\mathcal{D}$ of input facts, the \emph{Herbrand universe} $\Huniv$ is the
set of constants in $P \cup \mathcal{D}$ and the \emph{Herbrand base}
$\Hbase$ is the set of ground atoms over $\Huniv$. A program is a set of
definite (Horn) clauses and has a unique \emph{minimal model}, the least
Herbrand model, $\Mod_P(\mathcal{D}) \subseteq \Hbase$: the smallest set such that
\begin{enumerate}
\item $\mathcal{D} \subseteq \Mod_P(\mathcal{D})$ and
\item for every rule and every substitution $\theta$ from the rule's variables to $\Huniv$, if
$\{b_1\theta, \ldots, b_n\theta\} \subseteq \Mod_P(\mathcal{D})$ then
$h\theta \in \Mod_P(\mathcal{D})$.
\end{enumerate} The minimal model is the least fixpoint of the
immediate-consequence operator $T_P$, which exists by Knaster--Tarski since $T_P$
is monotone~\cite{tarski1955}. This least-model reading of a
definite program is standard~\cite{vanEmdenKowalski1976,ahv1995foundations}.

\subsection{Bottom-up vs.\ Top-down Evaluation}

\noindent\textbf{Top-down.} Given a \emph{query} such as
$\rel{Reach}(a,d)$?, evaluation searches backward: it finds rules whose head
unifies with the query, recursively proves the body, and backtracks on failure.

\noindent\textbf{Bottom-up.} Evaluation starts from $\mathcal{D}$, applies all
rules to derive new facts, and repeats to a fixpoint. For
$\mathcal{D}_0 = \{\rel{Edge}(a,b), \rel{Edge}(b,c), \rel{Edge}(c,d)\}$:
Iteration 1 (rule 1) derives $\rel{Reach}(a,b), \rel{Reach}(b,c), \rel{Reach}(c,d)$.
Iteration 2 (rule 2) adds $\rel{Reach}(a,c), \rel{Reach}(b,d)$, iteration 3 adds
$\rel{Reach}(a,d)$, and iteration 4 adds nothing (fixpoint). $\Mod_P(\mathcal{D}_0)$
contains these $6$ facts and the $3$ edges, independently of any query.

\subsection{2-local Ising Models}
\label{sec:ising}

An Ising model over $n$ binary spins
$\sigma = (\sigma_1, \ldots, \sigma_n) \in \{-1,+1\}^n$ has Hamiltonian
\[
H(\sigma) = \sum_{i=1}^n \eta_i\, \sigma_i + \sum_{i<j} J_{ij}\, \sigma_i \sigma_j,
\]
with \emph{local fields} $\eta \in \mathbb{R}^n$ and \emph{couplings}
$J \in \mathbb{R}^{n \times n}$, realized in hardware as programmable couplers. The \emph{ground state}
$\sigma^* = \argmin_\sigma H(\sigma)$ has minimal energy. Here \emph{2-local} means that $H$ has at most
pairwise interactions, a bound on the degree of $H$. 

A constraint coupling three or more spins at once has no
direct 2-local form, so encoding it requires auxiliary variables, which
Section~\ref{sec:framework} introduces
at the rule level and at the energy level. The induced interaction graph is dense, and its
connectivity does not in general match the sparse hardware graph of bounded degree (Chimera,
Pegasus, Zephyr~\cite{boothby2020pegasus,boothby2021zephyr}), so it is realized by
\emph{minor-embedding}, each logical spin becoming a \emph{chain} of physical qubits coupled
to act as one~\cite{choi2008}. We use the
spin encoding $\sigma_f = +1$ for \emph{true} and $\sigma_f = -1$ for
\emph{false}.

\section{Compilation Framework}\label{sec:framework}

We realize our compiler as four typed stages, each a transformation between
categories of decreasing abstraction:
\[
\mathcal{P} \xrightarrow{\ \mathcal{C}_1\ } \mathcal{P}^{\mathrm{bin}}
\xrightarrow{\ \mathcal{C}_2\ } \mathcal{P}^{\mathrm{gnd}}
\xrightarrow{\ \mathcal{C}_3\ } \Phi
\xrightarrow{\ \mathcal{C}_4\ } (\eta, J).
\]
A program $P \in \mathcal{P}$ binarizes to a program $\mathcal{C}_1(P) \in
\mathcal{P}^{\mathrm{bin}}$ whose rule bodies hold at most two atoms. Grounding sends that to
a variable-free program in $\mathcal{P}^{\mathrm{gnd}}$, Stage~3 reads off a propositional
formula $\Phi$ in conjunctive normal form (CNF), and Stage~4 emits the local fields $\eta$ and
couplings $J$ of a 2-local Ising model. Every object denotes the same minimal Herbrand
model $\Mod_P(\mathcal{D})$, and we design each stage to preserve that denotation.

We develop the four stages in order, each with its preservation lemma, and state two
correctness theorems. We illustrate the stages with the transitive closure program and
$\mathcal{D}_0$.

\subsection{Stage 1: Binarization
($\mathcal{C}_1\colon \mathcal{P} \to \mathcal{P}^{\mathrm{bin}}$)}

This stage rewrites rules with $n > 2$ body atoms into rules with at most $2$ body
atoms by introducing fresh auxiliary predicates. For $r = (h, b_1, \ldots, b_n)$,
fold the body right-to-left:
\[
\begin{aligned}
\mathcal{B}(r) = {}& \{(\aux_1(V_1), b_{n-1}, b_n)\} \\
&\cup\; \{(\aux_k(V_k), b_{n-k}, \aux_{k-1}(V_{k-1})) \mid 2 \le k \le n-2\} \\
&\cup\; \{(h, b_1, \aux_{n-2}(V_{n-2}))\}
\end{aligned}
\]
where $\aux_k$ summarises $b_{n-k} \wedge \cdots \wedge b_n$ and carries exactly
the still-needed variables, with $\mathit{vars}(\cdot)$ the set of variables occurring in its
argument:
\[
V_k = \mathit{vars}(\{b_{n-k}, \ldots, b_n\}) \cap
\big(\mathit{vars}(h) \cup \mathit{vars}(\{b_1, \ldots, b_{n-k-1}\})\big),
\]
preserving variable dependencies and rule safety (every head variable still occurs in the body).
Each step of the fold is a \emph{split}: it replaces the two-atom tail block by one fresh
auxiliary atom, so a rule with $n$ body atoms takes $n - 2$ splits.

\noindent\textbf{Example.} Both rules are already binary, so $\mathcal{C}_1(P_0) = P_0$,
and the two splits of
$\rel{Q}(X,W) \leftarrow \rel{S}(X,Y), \rel{T}(Y,Z), \rel{U}(Z,W), \rel{V}(W)$ give
{\footnotesize
\[
\aux_1(Z,W) \leftarrow \rel{U}(Z,W), \rel{V}(W)
\qquad \aux_2(Y,W) \leftarrow \rel{T}(Y,Z), \aux_1(Z,W)
\qquad \rel{Q}(X,W) \leftarrow \rel{S}(X,Y), \aux_2(Y,W).
\]}

\begin{lemma}[Binarization Correctness]\label{lem:bin}
For any program $P$ and EDB $\mathcal{D}$ whose own predicates avoid the auxiliary names,
given one fresh auxiliary predicate
name per split, restricted to the non-auxiliary
predicates, $\Mod_P(\mathcal{D}) = \Mod_{\mathcal{C}_1(P)}(\mathcal{D})$.
\end{lemma}
\begin{proof}[sketch]
Each split replaces the current two-atom tail block by a fresh atom $\aux(V_k)$, where $V_k$ are
the block variables that also occur elsewhere in the rule or in the head, and adds the single
rule deriving $\aux(V_k)$ from that block. The proof turns on the soundness of this
projection: a grounding satisfies the original body exactly when some extension over the
block-local variables makes $\aux(V_k)$ derivable, which is their existential elimination. The
auxiliary predicate is fresh, so
the added rules derive only auxiliary atoms and the step is a conservative extension. Induction
over the fold, one split per step, carries the equality to the
fully binarized program, and the minimal model is unchanged on the non-auxiliary predicates.
\end{proof}

\subsection{Stage 2: Grounding
($\mathcal{C}_2\colon \mathcal{P}^{\mathrm{bin}} \to
\mathcal{P}^{\mathrm{gnd}}$)}

This stage instantiates variables to constants in $\Huniv$. The atoms that grounding can
ever need form the \emph{Herbrand base subset}, or \emph{relevant base},
\[
\Hsub = \mathcal{D} \cup \{\, p(\vec{c}) \mid p \text{ a head predicate of } P,\ \vec{c} \in
\Huniv^{\mathit{ar}(p)} \,\},
\]
the EDB facts together with every ground atom of an intensional (head) predicate. An
extensional predicate contributes only its facts. An intensional predicate is instantiated over
$\Huniv$, so $|\Hsub|$ is $|\Huniv|^{\Theta(\mathit{ar})}$ whether or not the atoms are
derivable, the standard combined-complexity blowup~\cite{vardi1982} that demand-driven
grounding in the magic-sets tradition~\cite{bancilhon1986} would control. For each rule $r$ we
keep the \emph{relevant}
ground instances
\[
\begin{aligned}
\mathcal{G}(r) = \{\, r\theta \mid{} & \theta\colon \mathit{vars}(r) \to \Huniv,\ \text{every
body atom of } r\theta \text{ lies in } \Hsub \text{, and the head of } r\theta \text{ is not
among} \\
& \text{its body atoms} \,\},
\end{aligned}
\]
and emit each EDB fact as a bodiless ground rule. A ground instance with a body atom outside
$\Hsub$, for instance an extensional atom that is not a fact, can never fire and is omitted.
A ground instance whose head occurs in its own body fires only when its head already holds,
so it derives nothing and is omitted.

Restricting to $\Hsub$ is lossless. The decisive fact is what can \emph{never} appear:
every atom produced during bottom-up evaluation is an EDB fact or carries an intensional head
predicate, so it lies in $\Hsub$, and no atom outside $\Hsub$ is ever derived. A ground
instance discarded by the first condition has, by construction, a body atom outside $\Hsub$,
an atom no iterate can contain, so it
never fires in either the relevant or the full grounding.

\noindent\textbf{Example.} On $\mathcal{D}_0$ the Herbrand universe is $\{a,b,c,d\}$ and $\Hsub$
holds $19$ atoms (the $3$ edges and the $16$ ground $\rel{Reach}$ atoms). Grounding then
produces $18$ ground rules, the constraints over those $19$ atoms: the $3$ EDB facts, the $3$
instances of $r_1$ (one per
edge), and the $12$ instances of $r_2$ (each edge $\rel{Edge}(x,y)$ paired with a target $z$),
among them
\[
\rel{Reach}(a,b) \leftarrow \rel{Edge}(a,b)
\qquad
\rel{Reach}(a,d) \leftarrow \rel{Edge}(a,b), \rel{Reach}(b,d).
\]
The unrestricted grounding $\{r\theta \mid \theta\colon \mathit{vars}(r) \to \Huniv\}$ would
instead list $|\Huniv|^{\#\mathit{vars}(r)}$ instances per rule, $16 + 64 + 3 = 83$ in all, almost
every one dropped by the two conditions of $\mathcal{G}(r)$.

\begin{lemma}[Grounding Correctness]\label{lem:gnd}
Grounding preserves the minimal model:
$\Mod_{\mathcal{C}_2(P^{\mathrm{bin}})}(\mathcal{D}) =
\Mod_{P^{\mathrm{bin}}}(\mathcal{D})$. Grounding introduces no auxiliary predicates,
so the equality holds on the full Herbrand base.
\end{lemma}
\begin{proof}[sketch]
In two steps. \emph{(i) Full grounding preserves models.} A ground instance
$r\theta$ is satisfied by an interpretation exactly when the binary rule $r$ is satisfied by that
interpretation under $\theta$, so an interpretation models the full grounding iff it models
$P^{\mathrm{bin}}$, and the two share every model, hence the least one.
\emph{(ii) Relevance loses nothing.} The relevant grounding
$\mathcal{C}_2(P^{\mathrm{bin}}) = \bigcup_r \mathcal{G}(r)$ and the full grounding generate
\emph{term-by-term identical} bottom-up chains: both stay inside $\Hsub$,
and on any interpretation $I$ contained in
$\Hsub$ their one-step closures $I \cup T(I)$ agree, writing $T$ for each grounding's
immediate-consequence operator. An instance present in the full but
absent from the relevant grounding is dropped for one of two reasons. Either a body atom lies
outside $\Hsub$, which such an interpretation cannot contain, so the instance never fires, or
the head sits among its own body atoms, so whatever the instance fires from $I$ was already in
$I$. Identical chains reach the identical least fixpoint. Composing (i) and (ii)
gives $\Mod_{\mathcal{C}_2(P^{\mathrm{bin}})}(\mathcal{D}) = \Mod_{P^{\mathrm{bin}}}(\mathcal{D})$.
No auxiliary predicates are introduced, so the equality holds on the full
Herbrand base.
\end{proof}

\subsection{Stage 3: Ground Program $\to$ Boolean Formula
($\mathcal{C}_3\colon \mathcal{P}^{\mathrm{gnd}} \to \mathrm{CNF}$)}

Input to this stage is the \emph{relevant} ground program
$\mathcal{P}^{\mathrm{gnd}} = \mathcal{C}_2(P^{\mathrm{bin}})$ of Stage~2 (the $18$-rule program
for the running example, not the unrestricted grounding), which Lemma~\ref{lem:gnd} certified has
the same minimal model as $P^{\mathrm{bin}}$. We record it as an explicit propositional formula.
Each ground rule $h \leftarrow b_1, \ldots, b_k$ becomes the Horn clause
$\neg b_1 \vee \cdots \vee \neg b_k \vee h$, a fact $h \leftarrow$ becomes the unit clause
$h$, and $\mathcal{C}_3(P)$ is their conjunction $\Phi$, a CNF over the ground atoms. Here the
meaning of the object shifts from derivation to satisfaction, as a clause carries polarities
and has no privileged head. An assignment satisfies $\Phi$ exactly when it models the ground
program, so the minimal model is the assignment that satisfies $\Phi$ while setting the
fewest atoms true. We therefore read the bottom-up answer as a \textbf{Min-Ones SAT}
instance.

Min-Ones SAT is NP-complete in general~\cite{kstw2001}, yet on Horn formulas the minimum
coincides with the unique least model, computable in linear
time~\cite{dowlinggallier1984}. What remains for
Stage~4 is to build an energy whose unique minimizer is that least model.

The lemma is well known, and we record it in the form the correspondence proof composes.

\begin{lemma}[Min-Ones SAT $=$ minimal model]\label{lem:minones}
For a ground program $P$ and EDB $\mathcal{D}$, the satisfying assignments of
$\Phi = \mathcal{C}_3(P)$ are exactly the models of $P$, and $\Mod_P(\mathcal{D})$ is the unique
Min-Ones SAT solution, since it satisfies $\Phi$ and is a proper subset of every other
satisfying assignment.
\end{lemma}
\begin{proof}[sketch]
A Horn clause $\neg b_1 \vee \cdots \vee \neg b_k \vee h$ holds at $I$ exactly when
$\{b_1, \ldots, b_k\} \subseteq I$ implies $h \in I$, so $I \models \Phi$ iff $I$ models $P$.
The minimal model is itself a model and lies inside every model, by the least-model
property~\cite{vanEmdenKowalski1976}. A proper subset of a finite set has
strictly smaller cardinality, which makes $\Mod_P(\mathcal{D})$ the unique Min-Ones
solution.
\end{proof}

\subsection{Stage 4: 2-local Ising Encoding
($\mathcal{C}_4\colon \Phi \to
\mathbb{R}^{V} \times \mathbb{R}^{V \times V}$)}
\label{sec:stage4}

We turn $\Phi$ into an energy whose unique minimizer is the least model that
Stage~3 identified, and this is the main technical step of the compiler. Each clause of
$\Phi$ is the Horn clause of a ground rule $h \leftarrow b_1, \ldots, b_k$ with $k \le 2$,
that is, the implication $(b_1 \wedge \cdots \wedge b_k) \to h$. Exactly one assignment of its
atoms \emph{violates} the rule: every body atom true and the head false. The rule's
\emph{gate} is an energy term over those atoms that takes one constant value, its
\emph{floor}, on every assignment satisfying the rule, and rises by a gap $\Delta > 0$ on the
violating assignment. Three gate shapes cover $k \le 2$. Table~\ref{tab:gates} lists their
energies, and, writing $s_x = \sigma_x$ for the spin of atom $x$ and $\eta_x$ for the local
field on that spin, the coefficients realizing them are:

\begin{description}
\item[\textsc{Fact} $(h\leftarrow)$, $0$-IMPLY:] $\eta_h := -5$, energy $-5\,s_h$
($\Delta = 10$). Forces $h$ true.
\item[$1$-IMPLY $(h\leftarrow b)$:] $\eta_b := +1$, $\eta_h := -1$, $J_{b,h} := -1$,
energy $s_b - s_h - s_b s_h$ ($\Delta = 4$). Pairwise, with no
extra spin needed.
\item[$2$-IMPLY $(h\leftarrow l, r)$:] the violation $l \wedge r \wedge \neg h$ is
a degree-3 term, lowered to a quadratic form with one \emph{ancilla} spin
$z$~\cite{rosenberg1975}: $\eta_z := +1$, $\eta_h := -1$, $J_{l,r} := +1$,
$J_{z,l} = J_{z,r} = J_{z,h} := -1$ ($\Delta = 4$, with $z$ minimised out).
\end{description}

\begin{table}[t]
\centering
\caption{Gate energies before scaling by $W$, the uniform weight that scales every gate
against the cardinality field. Each gate rests at its floor on every assignment
that satisfies its rule and rises by the gap $\Delta$ on the unique violating assignment,
marked $\ast$. The $2$-IMPLY energies have the ancilla $z$ minimized out.}
\label{tab:gates}
\begin{tabular}{cc}
\hline
\multicolumn{2}{c}{\textsc{Fact} $(h\leftarrow)$} \\
$h$ & $E$ \\
\hline
\textsc{f} & $+5\;\ast$ \\
\textsc{t} & $-5$ \\
\hline
\end{tabular}
\qquad
\begin{tabular}{ccc}
\hline
\multicolumn{3}{c}{$1$-IMPLY $(h\leftarrow b)$} \\
$b$ & $h$ & $E$ \\
\hline
\textsc{f} & \textsc{f} & $-1$ \\
\textsc{f} & \textsc{t} & $-1$ \\
\textsc{t} & \textsc{f} & $+3\;\ast$ \\
\textsc{t} & \textsc{t} & $-1$ \\
\hline
\end{tabular}
\qquad
\begin{tabular}{cccc}
\hline
\multicolumn{4}{c}{$2$-IMPLY $(h\leftarrow l, r)$} \\
$l$ & $r$ & $h$ & $E$ \\
\hline
\textsc{f} & \textsc{f} & \textsc{f} & $-2$ \\
\textsc{f} & \textsc{f} & \textsc{t} & $-2$ \\
\textsc{f} & \textsc{t} & \textsc{f} & $-2$ \\
\textsc{f} & \textsc{t} & \textsc{t} & $-2$ \\
\textsc{t} & \textsc{f} & \textsc{f} & $-2$ \\
\textsc{t} & \textsc{f} & \textsc{t} & $-2$ \\
\textsc{t} & \textsc{t} & \textsc{f} & $+2\;\ast$ \\
\textsc{t} & \textsc{t} & \textsc{t} & $-2$ \\
\hline
\end{tabular}
\end{table}

\noindent\textbf{Disjunction.} \looseness=-1 A head derivable by several rules contributes one gate per
rule on the shared head spin. Because $(a \to h) \wedge (b \to h)$ equals $(a \vee b) \to h$,
the encoding expresses disjunction without a dedicated construct.

\noindent\textbf{Minimality.} An implication gate binds the head only when the rule's body
holds, so a head with no satisfied body stays otherwise free. We add a uniform field
$\eta_f \mathrel{+}= \varepsilon/2$ with $\varepsilon > 0$ on every atom spin. A fact atom is
pinned by its gate, so its field term is a constant shift that moves no ground state. The
reference compiler implements exactly this energy.

As spins are $\pm 1$, a field $\eta_f s_f$ swings by $2\eta_f$ between true and false, and $\eta_f =
\varepsilon/2$ makes each true atom cost $\varepsilon$ relative to false. Among the models of
$P$, the objective becomes $\varepsilon$ times the number of true atoms. The minimal model is a proper subset of every other model (Lemma~\ref{lem:minones}),
so it alone has the fewest true atoms and the field selects it.

\noindent\textbf{Precision.} Accumulated field offsets must stay below the penalty of any single
violation, so we scale the gates by a weight $W$ and keep
$\varepsilon\, N < W$ for some $N$ bounding the minimal model's size. Each gate's gap is at
least $\Delta_{\min} = 4$ (Table~\ref{tab:gates}), so a violation costs at least $4W$. Since
$\Mod_P(\mathcal{D}) \subseteq \Hsub$, the static choice $N = |\Hsub|$ works and is known
before any evaluation. The compiler tightens it to $N = n_{\mathrm{atoms}}$, the atom-spin
count, which still bounds the minimal model because every minimal-model atom occurs in the
compiled program, and deploys $\varepsilon = 1$ with $W = n_{\mathrm{atoms}} + 1$, so
$N < W$ holds by construction. On every family evaluated here
$n_{\mathrm{atoms}} = |\Hsub|$. Relative to the gate
scale, the selecting field is then $\varepsilon/(W\,\Delta_{\min}) =
\Theta(1/n_{\mathrm{atoms}})$: it
shrinks as the atom count grows, and the evaluation quantifies what that costs on hardware.

\subsection{Correctness}

\begin{theorem}[Ising Encoding, $\varepsilon$-minimality]\label{thm:ising}
Let $H$ be the Hamiltonian of Stage~4 for a ground program $P$ at the deployed
$\varepsilon = 1$, under
$N < W$ for an integer weight $W$ and some $N \ge |\Mod_P(\mathcal{D})|$. Then $H$ has a
unique minimizing
\emph{true-set} (set of atoms assigned true), and the atom-projection of every ground state
equals the minimal model:
$\sigma_f = +1 \iff f \in \Mod_P(\mathcal{D})$.
\end{theorem}
\begin{proof}[sketch]
The model carries one atom spin per ground atom of $P$ and one private ancilla per
$2$-IMPLY gate.
For a true-set $I$ write $\mathrm{cost}(I) = \min_{z} H(\sigma_I, z)$ for the energy
with the atom spins fixed by $I$ and the ancillas minimized out, and let $c_0$ collect the
per-gate floors and the constant part $-n_{\mathrm{atoms}}/2$ of the uniform
field. The constant $c_0$ is independent of $I$, shifts every energy equally, and so
does not affect the argmin. No gate is violated exactly when $I$ models $P$, by
Lemma~\ref{lem:minones}, in which case every gate rests at its floor and
$\mathrm{cost}(I) = c_0 + |I|$. When $I$ fails to model $P$, some gate is violated
and $\mathrm{cost}(I) \ge c_0 + W$ (the physical rise is
at least $W\,\Delta_{\min}$). The minimal model then sits strictly below every non-model,
$\mathrm{cost}(\Mod_P(\mathcal{D})) = c_0 + |\Mod_P(\mathcal{D})| \le c_0 +
N < c_0 + W \le \mathrm{cost}(I)$, so no non-model can be a minimizer.

Among the models, the minimal model
is a proper subset of every other, so its cardinality term is strictly smaller and it is the
unique minimizing true-set. Each $2$-IMPLY ancilla takes whichever value minimizes its own gate.
On a gate whose head is true with exactly one true body atom the two values tie, so the full
configuration need not be unique, though the atom-projection is.
\end{proof}

\noindent\textbf{Foundedness on a cycle.} The uniform field is what rejects self-supporting
derivations, a behavior the acyclic running example cannot exhibit. Take the ground program
$p \leftarrow q$, $q \leftarrow p$ with empty EDB. It has two models, $\emptyset$ and
$\{p, q\}$. The second satisfies both rules, since each atom supports the other, yet it is
unfounded, resting on nothing outside the cycle, and the least model is $\emptyset$. The compiled
Hamiltonian, with the deployed $\varepsilon = 1$ and $W = N+1 = 3$, assigns energy $-7$ to
$\emptyset$ and $-5$ to $\{p, q\}$, so its unique ground state is $\emptyset$, ahead of
$\{p, q\}$ by exactly $2\varepsilon$, the
cardinality gap times the field (verified with \texttt{dimod.ExactSolver}).

\looseness=-1 The Clark completion~\cite{clark1978} reads each rule as a biconditional, here
$p \leftrightarrow q$, and still admits both models. The $\varepsilon$ field selects the
founded model directly, with no completion or
loop formulas~\cite{linzhao2004}. For a definite program the minimal model, the well-founded
model~\cite{vangelder1991}, and the unique stable model~\cite{gelfond1988} coincide, so on this
fragment the field reproduces those semantics rather than extending them.

Composing this encoding theorem with the three preservation lemmas carries the guarantee back
to the source program.

\begin{theorem}[Correspondence]\label{thm:correspondence}
Write $\mathcal{I} = \mathcal{C}_4 \circ \mathcal{C}_3 \circ \mathcal{C}_2 \circ
\mathcal{C}_1$. For every program $P$ and EDB $\mathcal{D}$ whose Herbrand universe is
nonempty, with the fresh-name supply of
Lemma~\ref{lem:bin}, if $\sigma$ is a ground state of
the Ising model $\mathcal{I}(P, \mathcal{D})$, then for every non-auxiliary ground atom $f$,
$\sigma_f = +1 \iff f \in \Mod_P(\mathcal{D})$. Theorem~\ref{thm:ising} applied to the
binarized program fixes the auxiliary atoms as well, each to its value in the least model of
the ground binarized program $\mathcal{C}_2(\mathcal{C}_1(P))$. Only the $2$-IMPLY ancilla spins are unconstrained off this
projection: a gate whose head is true with exactly one true body atom has two tying ancilla
values, so the full ground-state configuration need not be unique.
\end{theorem}
\begin{proof}[sketch]
\looseness=-1 The pipeline carries $\mathcal{D}$ inside the program as bodiless ground rules (the
Stage~2 convention), and a verified absorption lemma shows this lossless: the extended
program's least model over the empty EDB is exactly $\Mod_P(\mathcal{D})$. The deployed
weight $W = |\Hsub^{\mathrm{bin}}| + 1$, computed on the binarized ground program
($\Hsub^{\mathrm{bin}} = \Hsub$ whenever no rule needs binarization), meets the bound of
Theorem~\ref{thm:ising} with
$N = |\Hsub^{\mathrm{bin}}|$, so the theorem applies. The atom-projection of every ground state of
$\mathcal{I}(P)$ is then the minimal model of the ground program
$\mathcal{C}_2(\mathcal{C}_1(P))$, the Min-Ones SAT solution of Lemma~\ref{lem:minones}. That
model pins every ground atom of the binarized program, auxiliary atoms included, which is the
auxiliary-atom clause. Lemma~\ref{lem:gnd} and Lemma~\ref{lem:bin} carry it back through
grounding and binarization, so on the non-auxiliary predicates it is $\Mod_P(\mathcal{D})$.
\end{proof}

\noindent\textbf{Verification.} Table~\ref{tab:lean} maps the results of
this paper to their Lean identifiers, and the artifact (Lean development, notebooks, and
generated data) is hosted at
\url{https://github.com/brurucy/quantum-datalog}.\footnote{The repository will be made
public upon acceptance. A permanent DOI will accompany the camera-ready version.}

\begin{table}
\centering
\caption{Paper results and their Lean identifiers in the verified development.
Appendix~\ref{app:lean} reproduces the statements verbatim.}
\label{tab:lean}
\begin{tabular}{ll}
\hline
Result & Lean identifier \\
\hline
Lemma~\ref{lem:bin} (Binarization) & \hyperref[app:c1-correct]{\texttt{C1.correct}} \\
Lemma~\ref{lem:gnd} (Grounding) & \hyperref[app:c2-correct]{\texttt{C2.correct}} \\
Lemma~\ref{lem:minones} (Min-Ones SAT) & \hyperref[app:c3-sat]{\texttt{C3.sat}},
\hyperref[app:c3-minones]{\texttt{C3.minOnes}} \\
Theorem~\ref{thm:ising} ($\varepsilon$-minimality) & \hyperref[app:c4-encode]{\texttt{C4.encode}} \\
Theorem~\ref{thm:correspondence} (Correspondence) & \hyperref[app:c4-pipeline-correct-edb]{\texttt{C4.pipeline\_correct\_edb}} \\
$2$-locality of $H$ ($(\eta,J)$ extraction) & \hyperref[app:c4-twolocal]{\texttt{C4.twolocal}} \\
\hline
\end{tabular}
\end{table}

\section{Evaluation and Limitations}

We size the datasets and programs to what could fit a commercially available annealer as of
writing. The target is the hardware class, 2-local Ising machines. Concrete device parameters
come from D-Wave, whose Pegasus
(Advantage) and Zephyr (Advantage2) topologies have degree $15$ and $20$ and a few thousand
qubits~\cite{boothby2020pegasus,boothby2021zephyr}.
Correctness is settled by Theorem~\ref{thm:correspondence}, so the question here is whether an
instance of that scale could run. Table~\ref{tab:programs} defines every evaluated program.
$\rel{tc\_path}$ is the running example's program over a path EDB and $\rel{tc\_cycle}$ is the
same program over a cycle. $\rel{nonlinear\_tc\_path}$ replaces the recursive rule by one that
joins two derived atoms, so $\rel{Reach}$ atoms can support one another. The non-recursive
$\rel{linear\_chain}$ is a sparse contrast, and $\rel{reach}$ computes single-source
reachability over a hardware graph given as data, the family of the chip-scale experiments
below.

\begin{table}[t]
\centering
\caption{The evaluated programs. $n$ is the size parameter, and $\rel{reach}$ takes a
hardware graph $G = (V, E)$ and a seed node as data.}
\label{tab:programs}
\footnotesize
\begin{tabular}{lll}
\hline
Family & Rules & EDB $\mathcal{D}$ \\
\hline
$\rel{tc\_path}(n)$ &
\begin{tabular}[t]{@{}l@{}}
$\rel{Reach}(X,Y) \leftarrow \rel{Edge}(X,Y)$ \\
$\rel{Reach}(X,Z) \leftarrow \rel{Edge}(X,Y), \rel{Reach}(Y,Z)$
\end{tabular} &
$\rel{Edge}(i, i{+}1)$ for $0 \le i < n{-}1$ \\
\addlinespace
$\rel{tc\_cycle}(n)$ & as $\rel{tc\_path}$ &
$\rel{Edge}(i, (i{+}1) \bmod n)$ for $0 \le i < n$ \\
\addlinespace
$\rel{nonlinear\_tc\_path}(n)$ &
\begin{tabular}[t]{@{}l@{}}
$\rel{Reach}(X,Y) \leftarrow \rel{Edge}(X,Y)$ \\
$\rel{Reach}(X,Z) \leftarrow \rel{Reach}(X,Y), \rel{Reach}(Y,Z)$
\end{tabular} &
as $\rel{tc\_path}$ \\
\addlinespace
$\rel{linear\_chain}(n)$ &
$\rel{P}_i(X) \leftarrow \rel{P}_{i-1}(X)$ for $1 \le i < n$ &
$\rel{P}_0(c)$ \\
\addlinespace
$\rel{reach}(G)$ &
\begin{tabular}[t]{@{}l@{}}
$\rel{Reach}(X) \leftarrow \rel{Source}(X)$ \\
$\rel{Reach}(Y) \leftarrow \rel{Reach}(X), \rel{Edge}(X,Y)$ \\
$\rel{Reach}(X) \leftarrow \rel{Reach}(Y), \rel{Edge}(X,Y)$
\end{tabular} &
$\rel{Edge}(u,v)$ for $\{u,v\} \in E$, one $\rel{Source}$ \\
\hline
\end{tabular}
\end{table}

Instances are compiled with the reference compiler at the certified encoding's $\varepsilon = 1$
and $W = n_{\mathrm{atoms}} + 1$. Soundness needs only $|\Mod_P(\mathcal{D})| < W$
(Theorem~\ref{thm:ising}), met at every size. Every rule in Table~\ref{tab:programs} carries
at most two body atoms, so Stage~1's binarization never fires in these experiments, and the
auxiliary-atom path is exercised by the verified witness module instead. Experiments use \texttt{dwave-samplers} and
\texttt{minorminer}, with no QPU involved.

\noindent\textbf{The energy landscape.} \looseness=-1 Because bottom-up evaluation is tractable we know the
minimal model, and this instance is small enough to know everything else too. The program is
$\rel{nonlinear\_tc\_path}(4)$, whose recursive rule joins two derived atoms, over the path
$0 \to 1 \to 2 \to 3$, with $19$ ground atoms, so there are $2^{19} = 524{,}288$ atom
assignments. We compute the exact energy of every one of them, with each gate's ancilla at its
optimal value. The enumeration certifies that the minimal model is the unique global minimum
and that exactly $18$ assignments are single-flip local minima.

Figure~\ref{fig:landscape} walks through this space. Each position on the horizontal axis is
one assignment, consecutive assignments differ in exactly one atom, and the tick label names
the atom added at that step, with the curve interpolating between the plotted states. The
vertical axis is the assignment's energy. The walk starts at $\mathcal{D}$, the assignment
holding exactly the three $\rel{Edge}$ facts, which the $0$-IMPLY gates force into every
low-energy state. The next six steps add the derived $\rel{Reach}$ atoms in fixpoint order
through the iterates $T^1$ and $T^2$ (the first and second application of $T_P$ to
$\mathcal{D}$), arriving at the minimal model, the nine-atom assignment
at energy $-1{,}800.5$, the global minimum. The six steps after it add $\rel{Reach}(1,1)$,
$\rel{Reach}(2,1)$, $\rel{Reach}(2,2)$, $\rel{Reach}(3,1)$, $\rel{Reach}(3,2)$, and
$\rel{Reach}(3,3)$, an unfounded set whose atoms derive one another circularly. The resulting
assignment is a genuine local minimum a mere $6$ above the global one, one $\varepsilon$ per
unfounded atom, while leaving it costs at least $79$, since flipping any of the six unfounded
atoms back to false
violates a gate scaled to $W\Delta_{\min} = 80$. The humps between the two wells are the
partial states where a rule body is already satisfied but its head not yet added. These
foundedness traps are what energy minimization must avoid, and they set the recovery boundary
of the recursive family measured below.

\begin{figure}[htbp!]
\centering
\includegraphics[width=0.74\linewidth]{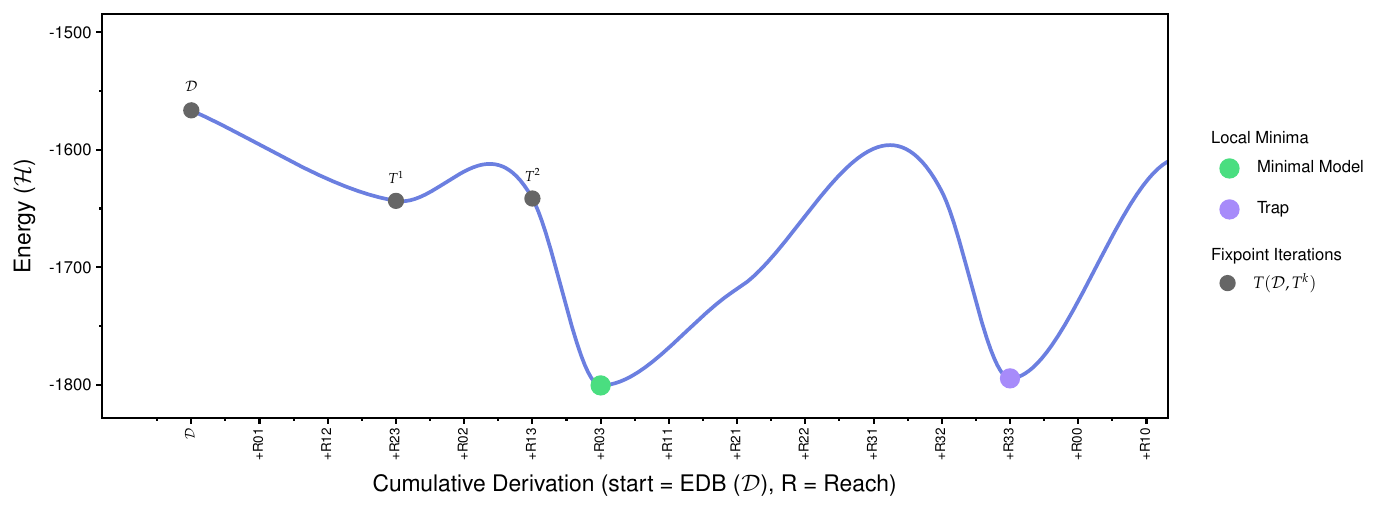}
\caption{Energy along a one-atom-per-step walk for $\rel{nonlinear\_tc\_path}(4)$: from the
EDB $\mathcal{D}$ through the fixpoint iterates $T^1, T^2$ (grey) to the minimal model
(green), the global minimum of all $2^{19}$ assignments by exhaustive enumeration. Adding the
six-atom unfounded set (purple) costs one $\varepsilon$ per atom, and leaving that local
minimum costs at least $79$ ($W\Delta_{\min} - \varepsilon$ on the gate scale).}
\label{fig:landscape}
\end{figure}

\noindent\textbf{Exact correspondence.} On the running example, the path
$a \to b \to c \to d$, the compiled model has $31$ spins, of which $19$ are atom
spins ($9$ true in the minimal model) and $12$ are ancillas, one per $2$-IMPLY
gate, with $51$ couplers. An exact solver's ground state projects exactly onto the
$9$-atom minimal model, and the cyclic, disjunctive, and multi-derivation cases are
instances of Theorem~\ref{thm:correspondence}'s guarantee.

\noindent\textbf{Logical to physical qubits.} The compiled model's logical interaction graph
must be minor-embedded into the annealer's sparse hardware graph
(Section~\ref{sec:ising}). We embed the compiled $\rel{tc\_path}$ family with
\texttt{minorminer} onto the Pegasus and Zephyr generators of \texttt{dwave\_networkx}.
Embedding is randomized, so we report medians over $15$ \texttt{minorminer} seeds
($0$--$14$, an embedding attempt being far cheaper than a recovery run). The running
example's $31$
logical spins become a median of $34$ physical qubits (range $33$--$36$), a $1.1\times$ overhead
with longest chain $2$.

Overhead climbs with size as the recursive join drives the logical
degree to $O(n)$, which passes Pegasus's $15$ at $n=8$ and Zephyr's $20$ at $n=10$. On the
Pegasus target, by $n=9$ the
$161$-spin model needs a median of $221$ qubits (range $201$--$253$), with a longest chain of
$7$ (range $5$--$9$). Longer chains also spend part of the energy scale on chain strength (the
coupling that holds a chain's qubits aligned), compounding the precision caveat below.

\noindent\textbf{Saturating the shipped machine.} Whole-chip topologies are the substrate of
D-Wave's beyond-classical simulations~\cite{king2025}, and a topology is itself a graph, so
it hosts a natural program: $\rel{reach}$ of
Table~\ref{tab:programs}, single-source reachability over the chip's own graph. The mirrored
third rule traverses stored edges backward, keeping one $\rel{Edge}$ fact per undirected
link. Because $\rel{Reach}$ is unary, the compiled model scales with the data:
$V + 3E + 1$ spins for a graph with $V$ nodes and $E$ edges, not $|\Huniv|^2$. The instance
whose EDB is the Zephyr Z1 graph compiles to $889$ logical spins and embeds into the
ideal Z12 generator ($4{,}800$ sites), the topology class of the shipped Advantage2, on
$1{,}696$ sites, with chains up to $26$ on the high-degree $\rel{Reach}$ hubs. Larger
single instances are out of reach. Z2 compiles to $3{,}833$ spins, above the machine at the
observed
$2\times$ embedding overhead, and Z3 exceeds it by spin count alone. The machine saturates
with independent programs instead: Figure~\ref{fig:embed} shows two separately compiled Z1
instances embedded side by side on $3{,}489$ sites. Compiled separately, each keeps its own
weight $W$, so the two models are independent Hamiltonians sharing only the chip, and one
anneal carries both. On a device the co-tenants would also share the single global coefficient
rescaling, harmless here since the two weights are equal.

\begin{figure}[tb]
\centering
\includegraphics[width=0.76\linewidth]{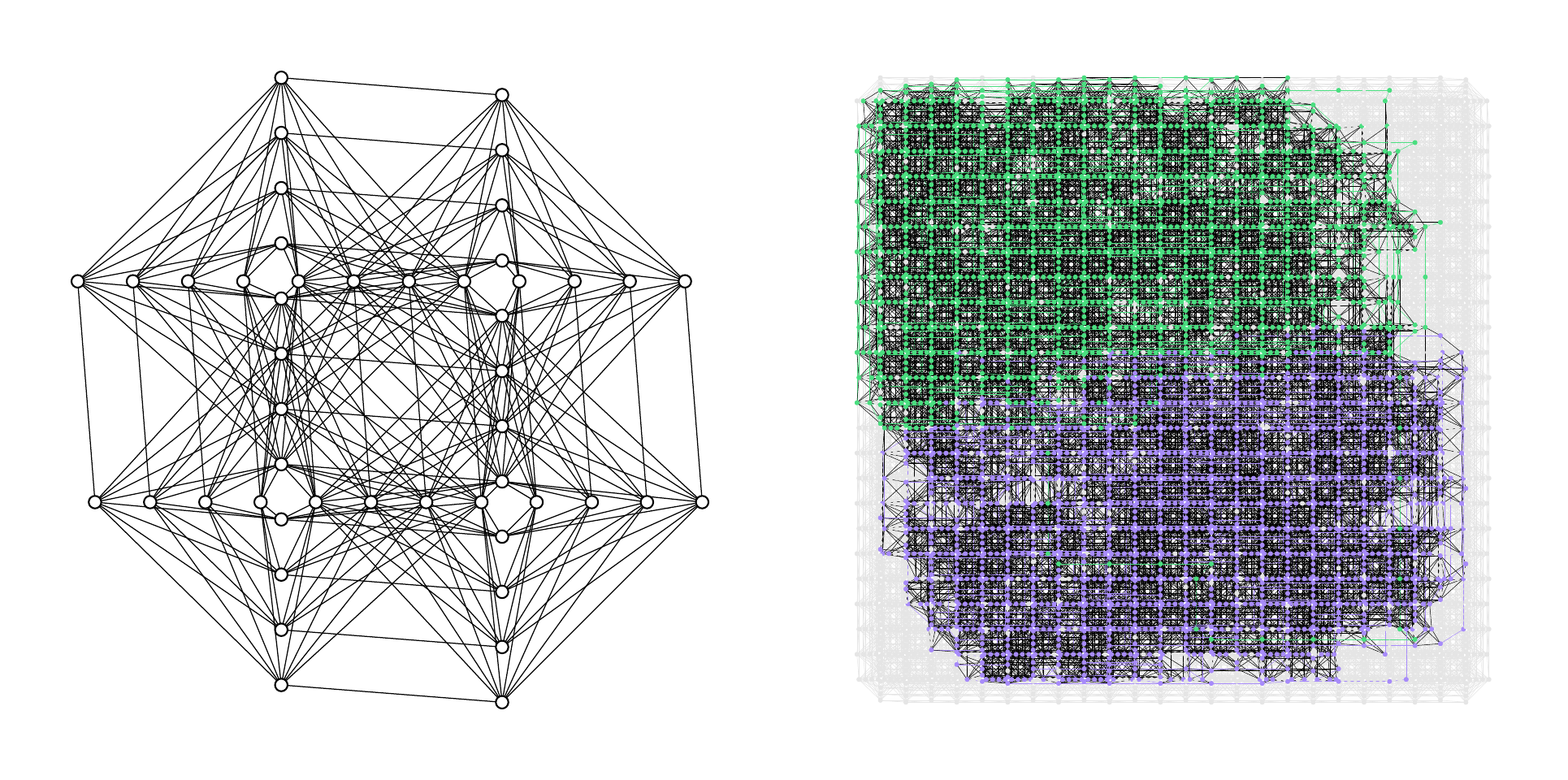}
\caption{Two independently compiled reachability programs at the shipped machine's scale.
Left: each instance's EDB, the Zephyr Z1 hardware
graph ($48$ nodes, $280$ edges). Right: both compiled programs, $1{,}778$ logical spins,
minor-embedded together into the ideal Zephyr Z12 generator ($4{,}800$ sites), occupying
$3{,}489$ physical qubits with longest chain $36$. One instance
is green, the other purple, unused qubits are greyed.}
\label{fig:embed}
\end{figure}

\noindent\textbf{Solver recovery.} Theorem~\ref{thm:correspondence} places the minimal model
at the unique ground state, and we measure whether a sampler reaches it. We compile each instance to its Ising model
and sample it with two solvers from D-Wave's SDK: classical simulated annealing (\textsc{sa},
\texttt{SimulatedAnnealingSampler}) and D-Wave's simulated quantum annealing (\textsc{sqa},
\texttt{PathIntegralAnnealingSampler}, a path-integral Monte Carlo surrogate for the device).
Both are classical simulations, so the recovery rates measured here do not bound a physical
anneal. Each runs at a budget of $200$ reads (independent anneals) of $1{,}000$ sweeps (update passes
over all spins), all other parameters at library defaults. We report the median over five seeds
($42$--$46$) with the
min--max range, all reproducible from the committed artifact. A run \emph{recovers} the model
when its lowest-energy read decodes to the bottom-up least model. The rate $p_{\text{succ}}$ is
the fraction of a run's reads that decode to it. Both solvers
carry full floating-point coefficients, isolating solver behavior from coupler
precision. 

Table~\ref{tab:recovery} reports $\rel{tc\_path}$: recovery is near-certain to $n=6$, then
collapses over $n=7\text{–}9$ ($p_{\text{succ}}$ from $0.76$ to $0.01$) and vanishes by $n=10$.
$\rel{tc\_cycle}$ breaks two steps earlier (its last recovery is at $n=7$), and the
non-recursive $\rel{linear\_chain}$ recovers
throughout the sweep to $n=30$. The non-linear $\rel{nonlinear\_tc\_path}$ collapses far sooner:
recovery is near-certain at $n=2$ but already gone by $n=4$ ($p_{\text{succ}}$
$0.97 \to 0.01 \to 0$), because its unfounded self-supporting sets are the deep
traps of Figure~\ref{fig:landscape}, basins a sampler falls into and rarely climbs out of.

At logical degree $9$ in both,
$\rel{nonlinear\_tc\_path}(3)$ with $11$ atoms recovers at $0.01$ while $\rel{tc\_path}(4)$
with $19$ atoms recovers at $0.99$, so size alone does not explain the gap. An acyclic linear
program admits no self-supporting set. Its shallow incomplete-derivation traps multiply with
$n$ and, with the climbing degree, plausibly drive $\rel{tc\_path}$'s later collapse. In these
families the wall's arrival tracks foundedness rather than coupler precision: both solvers
carry exact coefficients, yet the energy landscape hides the program's polynomial-time Horn
structure from local search, and more so as the rules admit more unfounded support.

\looseness=-1 The wall is calibrated to the budget rather than absolute: $10\times$ and $50\times$ the
sweeps lift \textsc{sa} on $\rel{tc\_path}(9)$ from $0.01$ to $0.175$ and $0.55$ (and $n=8$
to $0.985$, medians over the same five seeds), so the
crossing points shift with effort while the ordering in size and family persists.
\textsc{sqa} is never better than \textsc{sa} beyond
noise, and is worse on $\rel{linear\_chain}$ ($0.52$ against $0.97$ at $n=30$).

At the machine's scale, the $\rel{reach}$ program of
Table~\ref{tab:programs} over two disjoint copies of the Z1 graph, with a single
$\rel{Source}$ in the first copy, compiles to $1{,}777$ spins, and its minimal model lights
$\rel{Reach}$ on the sourced copy only. The count is one spin fewer than
Figure~\ref{fig:embed}'s pair of independent programs, whose second copy carries its own
$\rel{Source}$ atom. Under the protocol above neither solver recovers it
on any read ($p_{\text{succ}}$ $0.00$ [$0.00, 0.00$] over five seeds). Every failing read
decodes to the same assignment, the minimal model plus all $48$ unreached $\rel{Reach}$
atoms, an unfounded set in which the mirrored rule lets the endpoints of every stored edge
support each other. That assignment sits $48\varepsilon$ above the ground state, escaping it
requires climbing a single-violation barrier of $W\Delta_{\min} = 2{,}632$, and it survives
ten thousand times the sweep budget unchanged. The instance is classically trivial: an
incremental Datalog engine, PyDBSP 2.1.0~\cite{pydbsp2024}, computes the same minimal model
in under $7$ milliseconds (median of five runs on an Apple M5 Max laptop CPU).

\begin{table}[t]
\centering
\caption{Solver recovery on $\rel{tc\_path}$ ($200$ reads of $1{,}000$ sweeps): atom spins,
minimal-model size, maximum logical degree, and the recovery rate $p_{\text{succ}}$, median
over five seeds \,{\scriptsize[min, max]}, under \textsc{sa} and \textsc{sqa}.}
\label{tab:recovery}
\begin{tabular}{rrrrcc}
\hline
$n$ & atoms & $|\Mod|$ & degree & $p_{\text{succ}}$ \textsc{sa} & $p_{\text{succ}}$ \textsc{sqa} \\
\hline
\input{figures/recovery_rows}\end{tabular}
\end{table}

\noindent\textbf{Precision on a physical device.} The solvers above compute with exact
coefficients, while a QPU rescales them into a fixed analog range where integrated control
errors perturb each programmed value by about $10^{-2}$ at one standard
deviation~\cite{dwaveICE}. Resolving two coefficients reliably needs several standard
deviations of separation, leaving roughly $4$ to $5$ dependable bits. The smallest separation that matters is one atom of cardinality, worth
$\varepsilon = 1$, set against the largest programmed coefficient magnitude $c_{\max}$,
which the sparse sweep families measure at $4W$ (each fact's $-5W$ field gains $+W$ from its atom's
role as a $1$-IMPLY body), a relative gap
$\varepsilon / c_{\max} \sim 1/(4\,n_{\mathrm{atoms}})$. For $\rel{tc\_path}$ a $4$-bit
device already fails to resolve that gap at the smallest instance ($5$ atoms) and a $5$-bit
device at $11$ atoms, with the other sparse families crossing within a few atoms of the same
points. The
non-linear family's $c_{\max}/W$ grows with $n$, so its budget is strictly tighter.
\section{Conclusion}

\looseness=-1 We present
$\mathcal{I} = \mathcal{C}_4 \circ \mathcal{C}_3 \circ \mathcal{C}_2 \circ
\mathcal{C}_1$,
our compiler from recursive Datalog programs to 2-local Ising models, whose ground state
projects onto the least Herbrand model of the source program. The per-stage results carry the minimal model
through (Lemmas~\ref{lem:bin}--\ref{lem:minones} and
Theorem~\ref{thm:ising}), and their composition is the ground-state
correspondence of Theorem~\ref{thm:correspondence}. Empirically, recovery under the simulated
solvers tracks foundedness rather than instance size, failing completely at chip scale while
a classical engine answers in milliseconds. The contribution is the verified
bridge itself, and its grounding runs classically. Any quantum advantage would have to come from
NP-hard extensions of the source language, from distributing one program across several
annealers, or from a richer use of the annealer's quantum dynamics than a ground-state read
exploits.

%% Define the bibliography file to be used
\bibliography{main}\label{end:refs}

\clearpage
\input{appendix_lean}

\end{document}

%% file: figures/recovery_rows.tex
$4$ & $19$ & $9$ & $9$ & 0.99 {\scriptsize[0.98,1.00]} & 0.98 {\scriptsize[0.98,0.99]} \\
$5$ & $29$ & $14$ & $11$ & 0.99 {\scriptsize[0.98,0.99]} & 0.98 {\scriptsize[0.98,1.00]} \\
$6$ & $41$ & $20$ & $13$ & 0.98 {\scriptsize[0.98,1.00]} & 0.96 {\scriptsize[0.94,0.98]} \\
$7$ & $55$ & $27$ & $15$ & 0.76 {\scriptsize[0.73,0.78]} & 0.60 {\scriptsize[0.58,0.64]} \\
$8$ & $71$ & $35$ & $17$ & 0.20 {\scriptsize[0.17,0.28]} & 0.12 {\scriptsize[0.08,0.15]} \\
$9$ & $89$ & $44$ & $19$ & 0.01 {\scriptsize[0.01,0.03]} & 0.01 {\scriptsize[0.01,0.01]} \\
$10$ & $109$ & $54$ & $21$ & 0.00 {\scriptsize[0.00,0.00]} & 0.00 {\scriptsize[0.00,0.01]} \\
\hline

%% file: appendix_lean.tex
\appendix

\section{The Verified Statements}\label{app:lean}

The statements below are reproduced \emph{verbatim} from the Lean~4 development, generated
mechanically from the sources, with proofs elided as \texttt{...}. The namespace of each
identifier is its stage file, while the semantic layer (\texttt{Datalog.lean}) and the
witness module (\texttt{Nonvacuity.lean}) export top-level names. The development is
\texttt{sorry}-free and builds with mathlib. For every statement,
\texttt{\#print axioms} reports
exactly the three standard foundational axioms \texttt{propext}, \texttt{Classical.choice}, and
\texttt{Quot.sound}. The stage statements evaluate \texttt{minimalModel} at the empty
interpretation because Stage~2 stores each EDB fact as a bodiless rule
(\texttt{edbRules}). The absorption lemma
\hyperref[app:datalog-minimalmodel-edb-absorb]{\texttt{minimalModel\_edb\_absorb}}, printed
below, proves that
convention lossless, and
\hyperref[app:c4-pipeline-correct-edb]{\texttt{C4.pipeline\_correct\_edb}} composes it with
the empty-EDB core, so the conclusion speaks about \texttt{minimalModel P D} for an explicit
EDB \texttt{D}, matching Theorem~\ref{thm:correspondence} as stated.
Lemma~\ref{lem:bin}'s \texttt{C1.correct} already carries a
general interpretation \texttt{E}, which the composition specializes. The hypothesis
\texttt{hW} of \hyperref[app:c4-encode]{\texttt{C4.encode}} is exactly the bound
$N < W$ of Theorem~\ref{thm:ising}, and
\hyperref[app:c4-pipeline-correct-relevant]{\texttt{C4.pipeline\_correct\_relevant}}
discharges it with the static weight
$W = |\Hsub| + 1$ of Section~\ref{sec:stage4}. The certified model indexes its spins by every
ground atom and every clause. The energy touches only the formula's clauses, so the remaining
ancilla coordinates carry no terms, and an atom outside the formula carries only its
$\varepsilon$ field, which pins it false. The compiled model is the restriction to the
coordinates that occur, and the restriction changes no term and no ground-state atom value.

The auxiliary-name supply of Lemma~1 and Theorem~2 is fuel-bounded, with one fresh name per
binarization split, and is constructively inhabited. The witness module
\texttt{Nonvacuity.lean}, part of the build, instantiates every result below on concrete
programs, so every hypothesis is discharged by an explicit term. One witness program's
binarization consumes a fresh name. Another has a nonempty minimal model, and its witness
proves the ground state sets both derived atoms true, so the correspondence is exercised in
both directions. The statements are conditional on a supplied ground state, and
\hyperref[app:nonvacuity-exists-groundstate]{\texttt{exists\_groundState}}, printed last,
shows one always exists as the argmin over the finite spin space.

\noindent\textbf{The definitions.} The statements quantify over the definitions
below, reproduced verbatim as well. \texttt{litSat} and \texttt{clauseSat} are literal- and
clause-level satisfaction in the usual sense. \texttt{minimalModel} is the intersection of
all models. \texttt{IsMinOnes} demands a strictly unique cardinality minimum. The $\pm 1$
spin semantics enters through \texttt{spinVal} inside \texttt{IsTwoLocal}.
\texttt{Signature}, \texttt{Atom}, and
\texttt{Term} formalize Section~\ref{sec:background}'s syntax over a fixed finite signature.
Every predicate has positive arity, so a ground atom always carries a constant and supplies
a grounding. The hypothesis \texttt{[Nonempty Grounding]} therefore appears only on the two
EDB-absorption statements, whose programs need not contain a ground atom.
Further,
\texttt{Atom.IsGround} holds when every argument is a constant, \texttt{Safe} is
Section~\ref{sec:framework}'s range restriction, and \texttt{groundTerm} sends a term to a
constant, reading
variables off the grounding. \texttt{excess} uses natural-number subtraction, which
truncates at zero, so facts and short rules contribute nothing.
\texttt{BinaryProgram} and \texttt{GroundBinaryProgram} are the rule shapes with at most two
body atoms, ungrounded and ground, each with a \texttt{toProgram} coercion back to
\texttt{Program}, and \texttt{edbFacts} and \texttt{idbPreds} collect a
ground binary program's bodiless heads and intensional predicates. The stage compilers
\texttt{C1.run} through \texttt{C4.run} are not reproduced, since the statements assert
their correctness.
\begin{lstlisting}
def GroundAtom := { a : Atom // a.IsGround }

abbrev HerbrandBase := GroundAtom

abbrev Interpretation := Finset HerbrandBase

structure Rule where
  head : Atom
  body : Finset Atom
  negBody : Finset Atom
  safe : Safe head body
  deriving DecidableEq

abbrev Program := Finset Rule

def Grounding := Signature.Var → Signature.Const

def applyAtom (θ : Grounding) (a : Atom) : GroundAtom :=
  ⟨⟨a.pred, fun i => Term.const (θ.groundTerm (a.terms i))⟩, fun _ => trivial⟩

def applyHead (θ : Grounding) (r : Rule) : GroundAtom := θ.applyAtom r.head

def Satisfies (I : Interpretation) (r : Rule) : Prop :=
  ∀ θ : Grounding, (∀ a ∈ r.body, θ.applyAtom a ∈ I) → θ.applyHead r ∈ I

def IsModel (P : Program) (E I : Interpretation) : Prop :=
  E ⊆ I ∧ ∀ r ∈ P, Satisfies I r

def minimalModel (P : Program) (E : Interpretation) : Interpretation :=
  Finset.univ.filter (fun g => ∀ I : Interpretation, IsModel P E I → g ∈ I)

def GroundAtom.toFact (d : GroundAtom) : Rule :=
  ⟨d.val, ∅, ∅, by ...⟩

def edbRules (D : Interpretation) : Program := D.image GroundAtom.toFact

structure AuxSupply (isAux : Signature.Pred → Bool) (K : ℕ) where
  atom : ℕ → List Signature.Var → Atom
  vars_eq : ∀ c, c < K → ∀ vs, (atom c vs).vars = vs.toFinset
  is_aux : ∀ c, c < K → ∀ vs, isAux (atom c vs).pred = true
  pred_inj : ∀ c c', c < K → c' < K → ∀ vs vs',
    (atom c vs).pred = (atom c' vs').pred → c = c'

def excess (Q : Program) : ℕ := ∑ ρ ∈ Q, (ρ.body.card - 2)

def restrictNonAux (isAux : Signature.Pred → Bool) (I : Interpretation) : Interpretation :=
  I.filter (fun g => isAux g.pred = false)

def relevantBase (P : GroundBinaryProgram) : Interpretation :=
  edbFacts P ∪ Finset.univ.filter (fun g : GroundAtom => g.pred ∈ idbPreds P)

abbrev Literal := GroundAtom × Bool

abbrev Clause := Finset Literal

abbrev CNF := Finset Clause

def Sat (I : Interpretation) (φ : CNF) : Prop := ∀ c ∈ φ, clauseSat I c

def IsMinOnes (φ : CNF) (I : Interpretation) : Prop :=
  Sat I φ ∧ ∀ J : Interpretation, Sat J φ → J ≠ I → I.card < J.card

abbrev Spin := Bool

def spinVal (b : Spin) : Int := if b then 1 else -1

structure Model (ι : Type) where
  energy : (ι → Spin) → Int

def IsGroundState {ι : Type} (m : Model ι) (s : ι → Spin) : Prop :=
  ∀ t : ι → Spin, m.energy s ≤ m.energy t

def IsTwoLocal (f : (ι → Spin) → ℚ) : Prop :=
  ∃ (c₀ : ℚ) (h : ι → ℚ) (J : ι → ι → ℚ), ∀ σ : ι → Spin,
    f σ = c₀ + (∑ i, h i * (spinVal (σ i) : ℚ))
      + (∑ i, ∑ j, J i j * (spinVal (σ i) : ℚ) * (spinVal (σ j) : ℚ))
\end{lstlisting}

\phantomsection\label{app:c1-correct}%
\noindent\textbf{Lemma~1 (Binarization) --- \texttt{C1.correct}.}
\begin{lstlisting}
variable {isAux : Signature.Pred → Bool}
variable {K : ℕ} (supply : AuxSupply isAux K)

theorem correct (E : Interpretation) (P : Program)
    (hP : ∀ ρ ∈ P, isAux ρ.head.pred = false ∧ ∀ a ∈ ρ.body, isAux a.pred = false)
    (hE : ∀ g ∈ E, isAux g.pred = false) (hK : excess P ≤ K) :
    restrictNonAux isAux (minimalModel (BinaryProgram.toProgram (run supply P hP hK)) E)
      = minimalModel P E := by ...
\end{lstlisting}

\phantomsection\label{app:c2-correct}%
\noindent\textbf{Lemma~2 (Grounding) --- \texttt{C2.correct}.}
\begin{lstlisting}
theorem correct (BP : BinaryProgram) :
    minimalModel (toProgram (run BP)) ∅ = minimalModel (BinaryProgram.toProgram BP) ∅ := ...
\end{lstlisting}

\phantomsection\label{app:c3-sat}%
\noindent\textbf{Lemma~3 (Min-Ones SAT), the satisfaction half --- \texttt{C3.sat}.}
\begin{lstlisting}
theorem sat (P : GroundBinaryProgram) (I : Interpretation) :
    Sat I (run P) ↔ IsModel (toProgram P) ∅ I := by ...
\end{lstlisting}

\phantomsection\label{app:c3-minones}%
\noindent\textbf{Lemma~3 (Min-Ones SAT), the minimality half --- \texttt{C3.minOnes}.}
\begin{lstlisting}
theorem minOnes (P : GroundBinaryProgram) :
    IsMinOnes (run P) (minimalModel (toProgram P) ∅) := by ...
\end{lstlisting}

\phantomsection\label{app:c4-encode}%
\noindent\textbf{Theorem~1 ($\varepsilon$-minimality) --- \texttt{C4.encode}.}
\begin{lstlisting}
theorem encode (P : GroundBinaryProgram) (N : ℕ) (W : Int)
    (hN : (minimalModel (toProgram P) ∅).card ≤ N) (hW : (N : Int) < W)
    (σ : GroundAtom ⊕ Clause → Spin) (hs : IsGroundState (run (C3.run P) W) σ) :
    ∀ g : GroundAtom, σ (Sum.inl g) = true ↔ g ∈ minimalModel (toProgram P) ∅ := by ...
\end{lstlisting}

\phantomsection\label{app:datalog-minimalmodel-edb-absorb}%
\noindent\textbf{EDB absorption (the facts-as-rules convention is lossless) --- \texttt{minimalModel\_edb\_absorb}.}
\begin{lstlisting}
theorem minimalModel_edb_absorb [Nonempty Grounding]
    (P : Program) (D : Interpretation) :
    minimalModel (P ∪ edbRules D) ∅ = minimalModel P D := by ...
\end{lstlisting}

\phantomsection\label{app:c4-pipeline-correct-edb}%
\noindent\textbf{Theorem~2 (Correspondence) --- \texttt{C4.pipeline\_correct\_edb}.}
\begin{lstlisting}
theorem pipeline_correct_edb {isAux : Signature.Pred → Bool} {K : ℕ}
    [Nonempty Grounding]
    (supply : AuxSupply isAux K) (P : Program) (D : Interpretation)
    (hP : ∀ ρ ∈ P ∪ edbRules D,
      isAux ρ.head.pred = false ∧ ∀ a ∈ ρ.body, isAux a.pred = false)
    (hK : C1.excess (P ∪ edbRules D) ≤ K)
    (N : ℕ) (W : Int)
    (hN : (minimalModel (toProgram (C2.run (C1.run supply (P ∪ edbRules D) hP hK))) ∅).card ≤ N)
    (hW : (N : Int) < W)
    (σ : GroundAtom ⊕ Clause → Spin)
    (hs : IsGroundState (run (C3.run (C2.run (C1.run supply (P ∪ edbRules D) hP hK))) W) σ)
    (g : GroundAtom) (hg : isAux g.pred = false) :
    σ (Sum.inl g) = true ↔ g ∈ minimalModel P D := by ...
\end{lstlisting}

\phantomsection\label{app:c4-pipeline-correct-relevant}%
\noindent\textbf{Theorem~2 at the certified relevant-base weight $W = |\Hsub| + 1$ --- \texttt{C4.pipeline\_correct\_relevant}.}
\begin{lstlisting}
theorem pipeline_correct_relevant {isAux : Signature.Pred → Bool} {K : ℕ}
    (supply : AuxSupply isAux K) (P : Program)
    (hP : ∀ ρ ∈ P, isAux ρ.head.pred = false ∧ ∀ a ∈ ρ.body, isAux a.pred = false)
    (hK : C1.excess P ≤ K)
    (σ : GroundAtom ⊕ Clause → Spin)
    (hs : IsGroundState (run (C3.run (C2.run (C1.run supply P hP hK)))
            (((C2.relevantBase (C2.run (C1.run supply P hP hK))).card : Int) + 1)) σ)
    (g : GroundAtom) (hg : isAux g.pred = false) :
    σ (Sum.inl g) = true ↔ g ∈ minimalModel P (∅ : Interpretation) := ...
\end{lstlisting}

\noindent The $2$-locality of the Hamiltonian is verified through an
explicit witness: the development exhibits local fields $\eta$ and a coupling matrix $J$ and
proves that the Stage-4 energy equals
$c_0 + \sum_i \eta_i\, \sigma_i + \sum_{i,j} J_{ij}\, \sigma_i \sigma_j$, the symmetric
full-matrix form of the Hamiltonian of Section~\ref{sec:ising} (with
$\sum_{i<j}(J_{ij}+J_{ji})$ folded into the upper triangle). The extracted coefficients are
rational, and the additive constant $c_0$ is independent of the spins, so it leaves the
ground state unchanged.

\phantomsection\label{app:c4-twolocal}%
\noindent\textbf{$2$-locality of $H$ --- \texttt{C4.twolocal}.}
\begin{lstlisting}
theorem twolocal (P : GroundBinaryProgram) (W : Int) :
    IsTwoLocal (fun σ => ((run (C3.run P) W).energy σ : ℚ)) := ...
\end{lstlisting}

\phantomsection\label{app:nonvacuity-exists-groundstate}%
\noindent\textbf{Ground states exist (finite argmin) --- \texttt{exists\_groundState}.}
\begin{lstlisting}
theorem exists_groundState {ι : Type} [Fintype ι] [DecidableEq ι] (m : C4.Model ι) :
    ∃ σ, IsGroundState m σ := by ...
\end{lstlisting}

\section{Solver Configuration}\label{app:protocol}

The recovery experiments pin \texttt{dwave-samplers} 1.7.0, \texttt{minorminer} 0.2.21,
\texttt{dimod} 0.12.21, \texttt{dwave-networkx} 0.8.19, and \texttt{pydbsp} 2.1.0.
\textsc{sa} (\texttt{SimulatedAnnealingSampler}) runs at library defaults: a geometric
$\beta$ schedule over a range the library derives from the instance's couplings and fields,
one sweep per $\beta$ value, Metropolis acceptance, random initial states, and a fixed
spin-update order. \textsc{sqa} (\texttt{PathIntegralAnnealingSampler}) runs at library
defaults as well: transverse field $\Gamma = 1$, the default driver and problem field
schedules, and one qubit per chain and per update.